\documentclass[pdflatex,sn-mathphys-num]{sn-jnl}

\usepackage{graphicx}%
\usepackage{multirow}%
\usepackage{amsmath,amssymb,amsfonts}%
\usepackage{amsthm}%
\usepackage{mathrsfs}%
\usepackage[title]{appendix}%
\usepackage{xcolor}%
\usepackage{textcomp}%
\usepackage{manyfoot}%
\usepackage[ruled,linesnumbered]{algorithm2e}
\usepackage{booktabs}%
\usepackage{algpseudocode}%
\usepackage{listings}%
\usepackage{subcaption}

\theoremstyle{thmstyleone}%
\newtheorem{theorem}{Theorem}
\newtheorem{proposition}[theorem]{Proposition}%

\theoremstyle{thmstyletwo}%

\theoremstyle{thmstylethree}%
\newtheorem{definition}{Definition}%

\begin{document}

\title[Article Title]{Distributed quantum algorithm for the dihedral hidden subgroup problem}


\author[1]{\fnm{Pengyu} \sur{Yang}}

\affil*[1]{\orgdiv{College of Computer and Cyber Security}, \orgname{Fujian Normal University}, \orgaddress{\city{Fujian}, \postcode{350117},  \country{China}}}

\author[2]{\fnm{Xin} \sur{Zhang}}

\affil[2]{\orgdiv{School of  Mathematics and Statistics}, \orgname{Fujian Normal University}, \orgaddress{\city{Fujian}, \postcode{350117},  \country{China}}}

\author*[1]{\fnm{Song} \sur{Lin}}\email{lins95@gmail.com}



\abstract{To address the issue of excessive quantum resource requirements in Kuperberg's algorithm for the dihedral hidden subgroup problem, this paper proposes a distributed algorithm based on the function decomposition. By splitting the original function into multiple subfunctions and distributing them to multiple quantum nodes for parallel processing, the algorithm significantly reduces the quantum circuit depth and qubit requirements for individual nodes. Theoretical analysis shows that when $n\gg t$ ($t$ is the number of quantum nodes), the time complexity of the distributed version is optimized from $2^{O(\sqrt{n})}$ (the traditional algorithm's complexity) to $2^{o(\sqrt{n-t})}$. Furthermore, we carried out the simulation on the Qiskit platform, and the accuracy of the algorithm is verified. Compared to the original algorithm, the distributed version not only reduces the influence of circuit depth and noise, but also improves the probability of measurement success.}




\maketitle

\section{Introduction}\label{sec1}
\par Quantum computing is widely considered a key technology to break through the limits of classical computing due to its exponential acceleration potential in specific problems, such as Shor's algorithm\cite{Shor1994}, Grover's algorithm\cite{Grover-QMHS-PRL-1997}, and so on\cite{Harrow-QLSE-PRL-2009,Biamonte-QML-Nature-2017}. However, the current mainstream noisy intermediate-scale quantum (NISQ) devices \cite{Preskill-QCNISQ-Quantum-2018,Greinert-ECFQT-Zenodo-2024} are limited by the number of qubits, connectivity and noise interference, which means that it is difficult to support the large-scale quantum algorithms. In order to solve these problems, distributed quantum computing \cite{Caleffi-QINet-ACM-2018} comes into being. It is a new architecture that combines distributed computing and quantum computing, allowing different quantum processor nodes to communicate and cooperate to complete computing tasks. This "divide and conquer" strategy is able to make full use of existing NISQ devices to achieve quantum speedup while maintaining the overall efficiency of the algorithm.
\par In 2021, J.Avron et al. \cite{Avron-QANR-PRA-2021} proposed a specific distributed computing scheme, which decomposed the calculated Boolean function into multiple sub-functions and ran these on different quantum devices. Specifically, J.Avron et al. split the Boolean functions calculated in Grover's algorithm, Simon's algorithm and Deutsch-Jozsa algorithm to design the corresponding distributed quantum computing scheme. In addition, related experiments show that they not only reduce the depth of circuits, but also reduce the noise significantly. In 2022, Tan et al. \cite{Tan-DQSP-PRA-2022} improved the distributed Simon's algorithm, which reduced the complexity to $O({n})$ and solved the problem that the above distributed Simon algorithm could not expand the nodes to more than $2$. Compared with the original Simon's algorithm, the circuit depth is reduced from $O(n)$ to $O(n-t)$ ($2^t$ is the number of nodes). Subsequently, the team improved the distributed Grover's algorithm \cite{Qiu-DGrover-TCS-2024}, which required a smaller number of qubits and had a linear advantage in time complexity. Based on the exact Grover algorithm and the distributed scheme of splitting the original function, Zhou et al. \cite{Zhou-DExGrover-FP-2023} proposed the exact distributed Grover algorithm, which, like the exact Grover algorithm, can theoretically find the target state with 100\%. It is worth noting that the actual circuit depth of the algorithm is $8(n\ \mod\ 2)+9$, which is smaller than the circuit depth of the original and modified Grover algorithm, respectively. In addition, due to the shallow depth of the circuit, it is more resistant to depolarizing channel noise than several other Grover's algorithms. The above algorithms all reflect that when the original function is easy to split or satisfies a certain paradigm, the corresponding distributed algorithm can be naturally developed, which provides a good idea for designing distributed quantum algorithms. However, current research mainly focuses on Boolean function problems, and its adaptability to problems with complex algebraic structures remains to be explored.
\par The dihedral hidden subgroup problem (DHSP) is one of the key challenges in the field of quantum computing, and its efficient solution is of great significance for cracking lattice-based cryptosystems \cite{Ettinger-HSSA-ARXIV-1999}. Its goal is to find the generator of a subgroup from a black box function that hides the subgroup of a dihedral group. In 1998, Mark Ettinger and Peter Høyer pointed out that the query complexity required to solve this problem on classical computers is exponential \cite{Regev-QCLP-SIAM-2004}, and then proposed a quantum algorithm. Although the complexity of the quantum algorithm is polynomial, it has to call $o(2^n)$ times to solve such problems, so the total complexity is $o(2^n)$. In 2005, Kuperberg \cite{Kuperberg-SDHS-SIAM-2005} proposed an algorithm to solve DHSP with $2^{O(\sqrt{n})}$ time and space complexity. The hidden shift problem is another computing problem, whose aim is to solve it efficiently by using the parallelism of quantum algorithm through the periodicity or displacement property of the function. It is worth mentioning that under certain conditions, DHSP can be reduced to hidden shift problem. In addition, Kuperberg has also mentioned that his algorithm can solve the hidden shift problem. But as the scale of the problem increases, the complexity still requires much more qubits and circuit depth than the NISQ device can carry. However, existing distributed quantum algorithms have not systematically solved the parallelization requirements of such complex algebraic problems. 
\par To solve the above problems, we propose a distributed Kuperberg algorithm based on function decomposition. By splitting the original function into multiple sub-functions and assigning them to independent quantum nodes, the algorithm realizes task parallelism and resource decentralization. At the same time, it combines the quantum sorting network to optimize the efficiency of cross-node communication. Theoretical analysis shows that when $n\gg t$, the time complexity of the distributed version is optimized to $2^{O(\sqrt{n})}$ ($t$ is the number of nodes). Moreover, the circuit depth of single node is significantly reduced. In addition, the experimental results on the Qiskit platform also verify the feasibility of the algorithm.
\par The rest of this article is organized as follows. In Sect.\ref{sec2}, we review the Kuperberg's algorithm and mainly describes the algorithm flow. In Sect. \ref{sec3}, the distributed Kuperberg algorithm is proposed and the related mathematical proofs are given. Furthermore, we have disigned the related quantum circuit implementations and the experimental simulation is completed in Sect. \ref{sec4}. Finally, a summary of this paper is given.

\section{Preliminaries}\label{sec2}
\subsection{Dihedral hidden subgroup problem}\label{subsec2_1}

\begin{definition}[Dihedral group\cite{MacLane-Algebra-AMS-2023}]
The dihedral group $D_N$ is a symmetric group of a regular polygon, with $2N$ elements.
\end{definition}

\par $D_N$ contains all the symmetry transformations of the positive $N$-edge, including rotational symmetry and reflection symmetry, where the rotation angle 2$\pi/N$. $D_N$ can be defined as the semi- direct product of the second-order group consisting of the $N$th-order cyclic group $Z_N$ with the self-isomorphic reflection $s:x\mapsto x^{-1}$ on $Z_N$. The generating element of the Nth-order cyclic group $Z_N$ is $r$, and $D_N$ can be generated by $r$ and $s$, i.e. $D_N\cong Z_N\times\{e,s\}$. the elements of $D_N$ can be uniquely represented as $r^xs^h$ , $0\leq x\leq N- 1$, $h= 0, 1$, with the relational equation: $r^N=s^2=$ $srsr=e$, with e being the unit element.

\begin{proposition}\label{pr_1}
The dihedral group \(D_{N} \cong Z_{N} \times \{ e,s\}\) is isomorphic to the
semi-direct product of two cyclic groups \(Z_{N}\) and \(Z_{2}\), i.e. \(D_{N} \cong Z_{N} \times Z_{2}\).
\end{proposition}

By proposition \ref{pr_1} one can denote the elements of dihedral group \(D_{N}\) by \((b,d)\), where \(b \in\)\{0,1\}, \(d \in \{ 0,1,2,\cdots,N - 1\}\). When \(b = 0\), call \((b,d)\) a rotation of the dihedral group, and when \(b = 1\), call \((b,d)\) a reflection of the dihedral group.
\begin{proposition}
If \(N\) is even, two
subgroups \(\{(0,2x),(1,2x)|x \in Z_{N/2}\}\) and \(\{(0,2x),(1,2x + 1)|x \in Z_{N/2}\}\) about the dihedral group \(D_{N}\) are isomorphic to \(D_{N/2}\).
\end{proposition}
The dihedral hidden subgroup problem is described as follows: given a function \(h:D_{N} \rightarrow R\), where \(R\) is any set. This function \(h\) is invariant on the set of chaperones of the subgroup \(H \subseteq D_{N}\) and has different values on different chaperones, i.e. \(\forall c_{1},c_{2} \in D_{N}\), \(h\left( c_{1} \right) = h\left( c_{2} \right) \Leftrightarrow c_{1}H = c_{2}H\). Then DHSP is to find the subgroup \(H\) about this function\(h\). In 1999, Ettinger and Hoyer \cite{Ettinger-HSSA-ARXIV-1999} showed that when the subgroups \(H = \{(0,0),(1,d)\}\), \(H\) are generated by reflections \((1,d)\), DHSP can be reduced to the problem of finding the slope of reflections\(d\) of the generating elements of the implied subgroup \(H\) of the dihedral group, \(0 \leq d \leq N - 1\). Therefore, designing an efficient algorithm to obtain the slope becomes the key to the solution of DHSP.

\subsection{Consistency of the hidden shift problem with the dihedral hidden subgroup problem}\label{subsec2_2}
\begin{definition}[Hidden Shift Problem]
Given a group \((G, + )\), an output set \(A\), and two one-shot functions\(f,g:G \rightarrow A\). Suppose there exists an unknown \(a \in G\) such that \(f(x) = g(x + a)\) is satisfied for all \(x \in G\). The goal of the hidden shift problem is to find the shift \(a\).
\end{definition}

In 2005, Childs et al.\cite{Childs-QGHSP-ARXIV-2005} investigated the general case of the above problem and showed that when the finite group \(G\) is \(Z_{N}\), the hidden shift problem of finding the unknown displacement \(a \in Z_{N}\) is equivalent to the dihedral hidden subgroup problem. Therefore, in this case \cite{Moore-PFSAH-SIAM-2007,Castryck-FAHS-PQCrypto-2021}, if there exists an algorithm that can efficiently solve the dihedral implicit subgroup problem, it can efficiently solve the hidden shift problem, and vice versa. Specifically, define the function \(h:D_{N} \rightarrow A\) on the dihedral group, and let the subgroup of the function \(h\)
be \(H = \{(0,0),\{ 1,d\}\}\) ,\(H\) is generated by the reflection \(\{ 1,d\}\) when \(h(0,x)\) is injective, and according to the properties of dihedral groups it is known that:
\begin{align}
h(0,x) = h(1,x + d)
\end{align}
Define a monomial function on two cyclic groups \(f:Z_{N} \rightarrow A\), \(g:Z_{N} \rightarrow A\), \(A\) as any set of outputs, \(N\)=\(2^{n}\). By Property 3.1, we can make \(f(x) = h(0,x)\) , \(g(x) = h(1,x)\) , then equation 1 is equivalent:
\begin{align}
f(x) = g(x + d)
\end{align}
In other words, the reflection slope $d$ of the generator of the dihedral hidden subgroup problem is the same as the displacement $a$ of the hidden shift problem.

\subsection{Kuperberg's algorithm}\label{subsec2_3}
According to (Section 2.2) Sect. \ref{subsec2_2}, the hidden subgroup problem solved by Kuperberg's algorithm can be equated to the hidden shift problem. Next, the Kuperberg's algorithm will be described from the perspective of solving the hidden shift problem, and the specific flow is shown in \ref{Algorithm 1}. First, assume that there exists an Oracle that can efficiently query the values of functions $f$ and $g$. The algorithm can be used to solve the hidden subgroup problem:
\begin{align}
|b\rangle|x\rangle|y\rangle\overset{o}{\operatorname*{\rightarrow}}
\begin{cases}
|0\rangle|x\rangle|y\oplus f(x)\rangle \ if \ b=0 \\
|1\rangle|x\rangle|y\oplus g(x)\rangle \ if \ b=1 & 
\end{cases}
\end{align}

\begin{algorithm}[!ht]
\setlength{\abovedisplayskip}{1.2pt}
\setlength{\belowdisplayskip}{1.2pt}
    \caption{Kuperberg's algorithm}\label{algorithm2}\label{Algorithm 1}
	\KwIn{positive integers \(N\)=\(2^{n}\), black boxes about functions \(f:Z_{N} \rightarrow R\), \(g:Z_{N} \rightarrow R\), \(R\) for any set.}
	\KwOut{positive integers \(N\)=\(2^{n}\), black boxes about functions \(f:Z_{N} \rightarrow R\), \(g:Z_{N} \rightarrow R\), \(R\) for any set.}
    step 1: Prepare \( |0\rangle|0^{\otimes n}\rangle|0^{\otimes m}\rangle \), apply \( H^{\otimes(n+1)} \) to the first and second registers:  
   \[|\psi\rangle = \frac{1}{\sqrt{2}}(|0\rangle + |1\rangle) \otimes \frac{1}{\sqrt{2^n}}\sum_{x \in \{0,1\}^n}|x\rangle|0\rangle.\]\\
    step 2: Query the oracle:  
    \[
    |\psi\rangle = \frac{1}{\sqrt{2^{n+1}}}\sum_{x \in \{0,1\}^n} \left(|0\rangle|x\rangle|f(x)\rangle + |1\rangle|x\rangle|g(x)\rangle\right).
    \]\\
    step 3: Measure the third register, collapsing to \( y_0 \)Collapse first and second registers to get:  
    \[
     |\psi\rangle = \frac{1}{\sqrt{2}}\left(|0\rangle|x_0\rangle + |1\rangle|x_0 + a\rangle\right).
     \]  \\
    step 4: Apply Quantum Fourier Transform (QFT) to the second register:
    \[|\psi\rangle=\frac{1}{\sqrt{2^{n+1}}}(\sum_{j=0}^{2^{n}-1}e^{2\pi i jx_{0}/2^{n}}|0\rangle|j\rangle+\sum_{k=0}^{2^{n}-1}e^{2\pi ik(x_{0}+a)/2^{n}}|1\rangle|k\rangle )\] \\
    step 5: Measure the second register to obtain \( l \), collapsing the first register to:
    \[
    \begin{aligned}
    |\psi_l\rangle &= \frac{1}{\sqrt{2}} \left( e^{2\pi ilx_0/2^n}|0\rangle + e^{2\pi il(x_0+a)/2^n}|1\rangle \right) \\
    &= \frac{1}{\sqrt{2}} e^{2\pi ilx_0/2^n} \left( |0\rangle + e^{2\pi ila/2^n}|1\rangle \right) \\
    &= |0\rangle + e^{2\pi ila/2^n}|1\rangle
    \end{aligned}\]\\
    step 6: \( |\psi_{2^{n-1}}\rangle \) is obtained by the sieve method proposed by Kuperberg, with $l$ equal to $2^{n-1}$, at which point: 
    \[|\psi_{2^{n-1}}\rangle=|0\rangle+e^{\pi ia}|1\rangle\]\\ 
    step 7: Apply a Hadamard gate, when the last bit of $a$ is an even measurement yields $0$ and the last bit is an odd measurement yields $1$.
\end{algorithm}

Obviously, a single run of Algorithm 1 yields the last bit $a_0$ of $a$. By \ref{Property 2}, the functions \(f^{'}(x) = f(2x)\) and \(g^{'}(x) = g(2x + a_{0})\) can be constructed. We running Algorithm 1 again on the basis of the new functions constructed gives the last bit of \(a^{'} = \frac{(a - a_{0})}{2}\) , which is the penultimate bit of $a$. Recursing this procedure yields the remaining bits of $a$ obtained.

It is worth noting that the sieving method used in step 6 of the algorithm is the most central part of the Kuperberg's algorithm flow, which can reduce the complexity of the algorithm to the sub-exponential level. This is because there are $2^n$ possibilities of l obtained by measurement in step 5. In order to obtain $l=2^{n-1}$, multiple measurements are needed and its time complexity is exponential, which will lose the advantage of quantum algorithm. In order to retain the advantages of quantum algorithms, Kuperberg proposed the sieving method, whose specific steps are described as follows:

Firstly, steps 1-5 of algorithm \ref{Algorithm 1} are used as black boxes to generate states \(|\left. \psi_{l} \right\rangle\). We set \(m\) nodes, \(m = \left\lceil \sqrt{n - 1} \right\rceil\). Through node 1 (measuring l have you got information), find out and quantum state \(l\) bits of \(|\psi_{l_{1}}\rangle\) lowest \(m\) bits of the same quantum state \(l\) bits of \(|\psi_{l_{2}}\rangle\), and the quantum state \(l\) bits of \(|\psi_{l_{1}}\rangle\) and \(|\psi_{l_{2}}\rangle\) execution with operation: CNOT was performed on \(|\psi_{l_{2}}\rangle\) with \(|\psi_{l_{1}}\rangle\) as the controlled bit, and then the second register was measured. At this time, the first register collapsed to \(|\left. \psi_{l_{1} \pm l_{2}} \right\rangle = |\left.  0 \right\rangle + e^{2\pi i(l_{1} \pm l_{2})a/2^{n}}|\left.  1 \right\rangle\). We get \(|\psi_{l^{'}}\rangle = |\left. \psi_{l_{1} - l_{2}} \right\rangle\) with \(50\%\) probability. Secondly, taking \(|\psi_{l^{'}}\rangle\) as the input state of node 2, after enough states are accumulated in node 2, the states with the same lowest \(m + 1,\cdots,2m\) bits are screened from these states, and the combined operation is performed again.The obtained quantum state is taken as the input to the next node, and so on, until the node \(m\) is executed, because \(m \times m = n - 1\), only two quantum states are left in the last node \(|\psi_{0}\rangle\)
and\(|\left. \psi_{2^{n - 1}} \right\rangle\). In other words, all the bits of the quantum state in the node are 0 except the most significant bit. Therefore, repeating this sieving step many times can obtain \(|\left.\psi_{2^{n - 1}} \right\rangle\)  with high probability.

Finally, the time complexity of Kuperberg's algorithm is
briefly analyzed. The complexity of this algorithm is determined by two
main components: first, the number of bits in the hidden shift \(a\) bit. From Property \ref{pr_1}, the number of bits of \(a\)
is \(n = \left\lceil \log_{2}{|Z_{N}|} \right\rceil\) , then Algorithm 1 needs to perform \(O(n)\) iterations in total. The second is the complexity of the sieving method to get \(|\left. \ \psi_{2^{n - 1}} \right\rangle\) . From the specific
steps of the sieving method, it can be seen that each node needs at
least 4 states in order to output 1 state to the next node, and the current node needs to exist at least \(2^{m}\) states in order to find two states that meet the combination conditions with a probability of
50\%. Therefore node 1 needs at
least \(8^{m} \times 2^{m} = 2^{O(\sqrt{n})}\) states to
get \(|\left. \psi_{2^{n - 1}} \right\rangle\) with high probability.
The total time complexity
is \(2^{O(\sqrt{n})} + 2^{O\left( \sqrt{n - 1} \right)}\ldots\ldots{+ 2}^{O\left( \sqrt{2} \right)} = 2^{O(\sqrt{n})}\)
and the space complexity is also \(2^{O(\sqrt{n})}\).

\section{Distributed Kuperberg's algorithm}\label{sec3}
In this section, the proposed distributed Kuperberg's algorithm is presented. First, assume that there are \(2^{t}\) distributed quantum computing nodes, and each node is denoted as \(Node\ w\), \(w \in \{ 0,1\}^{t}\). Secondly, the definitions and theorems related to the algorithm are given.

\begin{definition}
Let the original function $f,g:Z_{N} \rightarrow R$ satisfy $f(x) = g(x + a)$, where $a \in Z_{N}$ is a hidden shift. $Z_{N}$ is decomposed into the input space t prefix and $(n - t)$ a suffix, $u \in \{ 0,1\}^{n - t},w \in \{ 0,1\}^{t}$, define a function:
\begin{align}
f_{w}(u)=f(w||u),g_{w}(u)=g(w||u)
\end{align}
Where $w||u$ denotes the string concatenation operation. Each subfunction $f_w$ and $g_w$ is assigned to an independent quantum Node, $Node w$, which only needs to process $(n-t)$ bits of input.
\end{definition}

\begin{definition}
For all $u \in \{ 0,1\}^{n - t}$, there exist sets $H(u)$ and $R(u)$ containing the subfunction values generated by all nodes, respectively:
\begin{align}
H(u)=\{f_{w}(u)\mid w\in\{0,1\}^{t}\},R(u)=\{g_{w}(u)\mid w\in\{0,1\}^{t}\}
\end{align}
\end{definition}
Due to the injective property of $f$ and $g$, the elements of $H(u)$ and $R(u)$ do not repeat each other. However, directly measuring these sets cannot directly obtain the information of hidden shift a, for two reasons: (1) the calculation results of different nodes are independent of each other, and the correlation across nodes cannot be directly established. (2) The elements in the sets $H(u)$ and $R(u)$ do not establish an explicit correspondence with the hidden shift $a$. In order to establish the relationship between subfunction values and hidden shifts, it is necessary to sort the set elements globally. Define the sorted strings $F(u)$ and $G(u)$ as follows:
\begin{definition}
For all $u \in \{ 0,1\}^{n - t}$ , the sorted strings $F(u)$ and $G(u)$ are as follows.
\begin{equation}
\begin{aligned}
F(u) &= f(w_0||u)||f(w_1||u)\cdots||f(w_{2^t-1}||u) \in \{0,1\}^{2^t m}, \\
G(v) &= g(w_0||v)||g(w_1||v)\cdots||g(w_{2^t-1}||v) \in \{0,1\}^{2^t m}.
\end{aligned}
\end{equation}
$f(w_0||u)\leqslant f(w_1||u)\leqslant\cdots\leqslant f(w_{2^t-1}||u),g(w_0||v)\leqslant g(w_1||v)\leqslant\cdots\leqslant g(w_{2^t-1}||v),w_i\in\{0,1\}^t$ When $i \neq j$, $w_{i} \neq w_{j}$.
\end{definition}
The sorting operation ensures that the generation of $F(u)$ and $G(u)$ depends only on the function value itself and is independent of the computing node, thus eliminating the randomness of the data distribution between nodes.
\setcounter{theorem}{0} 
\begin{theorem}\label{Theorem 1}
Let \(a = a_{1} \parallel a_{2}\), where \(a_{1} \in \{ 0,1\}^{t}\), \(a_{2} \in \{ 0,1\}^{n - t}\), for all \(u \in \{ 0,1\}^{n - t}\), \(v \in \{ 0,1\}^{n - t}\), there
exists \(a_{2}\) such that:
\begin{align}
F(u)=G(v) \text{ if and only if } u+a_2=v
\end{align}
\end{theorem}
\begin{proof}
Due to the injectivity of $f$ and $g$, the elements in $H(u)$ and $R(u)$ are all distinct for \(\forall u,\upsilon \in \{ 0,1\}^{n - t}\). \(H(u) = R(v)\) if and only if \(F(u)\)=\(G(v)\), which is equivalent to proving that $F(u) = G(v)$ if and only if $u + a_2 = v$. Let $w + a_1 = w'$:
Necessity: $\forall z \in H(u)$,  $\exists w \in \{ 0,1\}^{t},z = f\left( w||u \right)$ . Since $f\left( w||u \right) = g\left( w||u + a \right),z = g\left( w||u + a \right)$. When $u + a_{2}$ has no carry, $z = g\left( w + a_{1}||u + a_{2} \right)=g\left( w^{'}||v \right) \in R(v)$. When $u + a_{2}$ has a carry, $z = g\left( w + a_{1}||u + a_{2} \right)=g\left( w^{'} + 1||v \right) \in R(v)$ . Therefore, $H(u) \subseteq R(v)$. By the same reasoning, $R(v) \subseteq H(u)$, so $F(u)=G(v)$.
Sufficiency: When $F(u) = G(v)$, $\forall z \in H(u),\exists w \in \{ 0,1\}^{t},z = f\left( w||u \right)=g\left( w^{'}||u + a \right)$. When $u + a_2$ has no carry, $w || u + a = w' || v, w + a_1 = w', u + a_2 = v$. When $u + a_2$ has a carry, $w || u + a = w' + 1 || v, u + a_2 = v$. In conclusion, $u + a_2 = v$.
\end{proof}
According to Theorem\ref{Theorem 1} the core of the distributed Kuperberg's algorithm is to obtain \(F(u)\) and \(G(v)\) , calculate the corresponding subfunctions through different nodes, and then use the quantum sorting algorithm to obtain \(F(u)\) and \(G(v)\) , and then extract the last bits of \(a_{2}\) through quantum Fourier transform (QFT). According to \textbf{property 2} we can modify the original
functions \({f_{w}}^{'}(u) = f_{w}(2u)\) and \({g_{w}}^{'}(u) = g_{w}(2u + a_{0})\), where \(a_{0}\) is the last bit of \(a_{2}\), run the algorithm again to get the remaining bits of \(a_{2}\), and then recursively get \(a_{1}\) by the above formula,
and finally recover all the information of \(a\). For easy understanding, the distributed Kuperberg's algorithm for two nodes, i.e. the case of  \(t = 1\), is firstly given as shown in Algorithm \ref{algorithm2}, and its corresponding circuit diagram is shown in Fig.\ref{fig1}.
\begin{algorithm}[!ht]
\setlength{\abovedisplayskip}{1.2pt}
\setlength{\belowdisplayskip}{1.2pt}
    \caption{Distributed Kuperberg's algorithm (2 nodes)}\label{algorithm2}
	
	\KwIn{the sub-function Oracles \({O_{f}}_{w}\) of the function \(f:Z_{N} \rightarrow R\) with the sub-function Oracles \({O_{g}}_{w}\) of \(g:Z_{N} \rightarrow R\), \(R\) for any set, and \(w \in \{ 0,1\}^{2}\).}
	\KwOut{Hide the last bit \(a_{0}\) of \(a_{2}\) in the shift \(a = a_{1}||a_{2}\), \(a_{1} \in \{ 0,1\}^{1}, a_{2} \in \{0,1\}^{n - 1}\).}
    step 1: Prepare the quantum state \(|0\rangle|0^{\bigotimes n - 1} \rangle\left.|0^{\bigotimes m} \right\rangle\left. |0^{\bigotimes m} \right\rangle\left. |0^{\bigotimes 2m} \right\rangle\) and apply \(H^{\bigotimes(n)}\) to registers 1 and 2:
    \[\frac{1}{\sqrt{2}}( |0\rangle + |1 \rangle)\frac{1}{\sqrt{2^{n - 1}}}\sum_{u \in {\{ 0,1\}}^{n - 1}}^{}{|u \rangle|  0 \rangle}| 0 \rangle|0 \rangle\]\\
    step 2: Query the oracle: 
    \[\frac{1}{\sqrt{2}}\frac{1}{\sqrt{2^{n - 1}}}\sum_{u \in {\{ 0,1\}}^{n - 1}}^{}{(|  0 \rangle | u \rangle\left| \left. f_{1}(u) \right\rangle|f_{0}(u) \right\rangle| 0 \rangle + \left|  1 \right\rangle | u \rangle | g_{1}(u) \rangle| g_{0}(u) \rangle| 0 \rangle)}\]
    step 3: Act on registers 3 and 4 at \(U_{sort}\) and store the result after
sorting registers 3 and 4 in register 5:  
    $$\frac{1}{\sqrt{2}}\frac{1}{\sqrt{2^n}}\sum_{u \in \{0,1\}^{n-1}}(|0\rangle|u\rangle|f_1(u)\rangle|f_0(u)\rangle|S(u)\rangle + |1\rangle|u\rangle|g_1(u)\rangle|g_0(u)\rangle|T(u)\rangle)$$  \\
    step 4: Measure register 5, collapsing register 2 to \(u_{0}\) and backing out registers 3 and 4 :
    $$\frac{1}{\sqrt{2}}(|0\rangle|u_0\rangle + |1\rangle|u_0 + a_2\rangle)$$ \\
        step 5: Apply QFT to register 2:
    $$\frac{1}{\sqrt{2^n}}(\sum_{j=0}^{2^{n-1}-1}e^{2\pi iju_0/2^{n-1}}|0\rangle|j\rangle + \sum_{k=0}^{2^{n-1}-1}e^{2\pi ik(u_0+a_2)/2^{n-1}}|1\rangle|k\rangle)$$\\

    step 6: Measure register 2 and get \(l\), Register 1 collapses to: 
    \[
    \begin{aligned}
    |\psi_l\rangle &= \frac{1}{\sqrt{2}} \left( e^{2\pi ilu_0/2^n}|0\rangle + e^{2\pi il(u_0+a_2)/2^n}|1\rangle \right) \\
    &= \frac{1}{\sqrt{2}} e^{2\pi ilu_0/2^n} \left( |0\rangle + e^{2\pi ila_2/2^n}|1\rangle \right) = |0\rangle + e^{2\pi ila_2/2^n}|1\rangle
    \end{aligned}\]\\
    step 7: Through the sieve method to get \(|\psi_{2^{n - 2}} \rangle\) , \(l\) is equal to \(2^{n - 2}\) , using the \(H\) door, when the last digit of \(a_{2}\) is an even number of measurements to get \(0\), The last digit is an odd
    number of measurements to get 1.
\end{algorithm}
\newpage
\begin{figure}[h]
\centering
\includegraphics[width=0.9\textwidth]{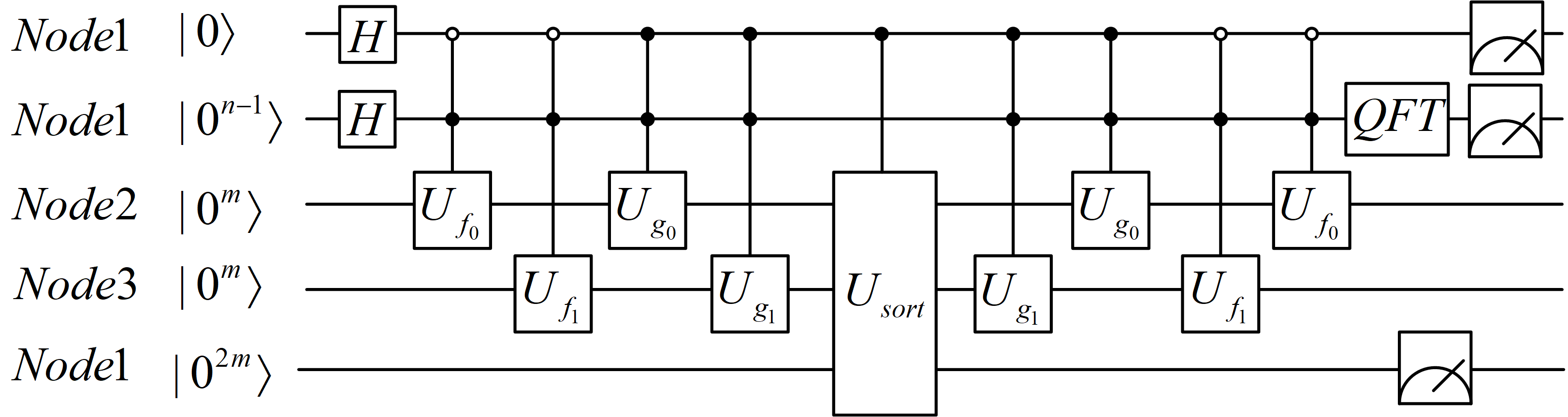}
\caption{Distributed Kuperberg algorithm (2-node)}\label{fig1}
\end{figure}
Where, \(U_{sort}\) implements the comparison of the values of two
registers, after which the result of sorting in dictionary order and the
third register are subjected to an all-or-nothing operation, with the
expression
\begin{equation}
\begin{aligned}
U_{sort}|q\rangle|p\rangle|r\rangle=
\begin{cases}
|q\rangle|p\rangle|r\otimes(q||p)\rangle,p\geq q \\
|q\rangle|p\rangle|r\otimes(p||q)\rangle,p\leq q & 
\end{cases},p,q\in\{0,1\}^{m},r\in\{0,1\}^{2m}
\end{aligned}
\end{equation}
Obviously, \(U_{sort}\) does not change the value of the comparison, but only does a controlled operation based on the result of the registers. A controlled quantum gate operation between two different nodes can be realized using quantum stealth transmutation\cite{Bennett-QT-PRL-1993}.
Further, Algorithm 2 can be extended to \(2^{t}\) nodes: replacing
registers 3 and 4 with \(2^{t}\) control registers corresponding
to \(2^{t}\) distributed computing nodes. \(U_{sort}\) The implementation
of the algorithm is changed to use the sorting network
\cite{Paterson-ISN-ALGO-1990} to perform a dictionary order sorting
operation on the \(2^{t}\) control registers, and the \(2^{t}\) elements can be sorted with the complexity of \(O(t)\). The specific algorithm
flow is described in Algorithm \ref{algorithm3}.

\begin{algorithm} [H]
\setlength{\abovedisplayskip}{1.2pt}
\setlength{\belowdisplayskip}{1.2pt}
    \caption{Distributed Kuperberg's algorithm ($t$ nodes)}\label{algorithm3}
	
	\KwIn{a positive integer \(N\)=\(2^{n}\) , with the number of
nodes \(T = 2^{t}\) , the sub-function Oracle \({O_{f}}_{w}\) of the
function \(f:Z_{N} \rightarrow R\) with the sub-function
Oracle \({O_{g}}_{w}\) of \(g:Z_{N} \rightarrow R\) , \(R\) for any
set, \(w \in \{ 0,1\}^{t}\).}
	\KwOut{Hide the last bit \(a_{0}\) of \(a_{2}\) in the shift \(a = a_{1}||a_{2}\) , \(a_{1} \in \{ 0,1\}^{t}, a_{2} \in \{0,1\}^{n - t}\).}
    step 1: Prepare the quantum state$|0\rangle|0\rangle|0^{\otimes m}\rangle|0^{\otimes m}\rangle\ldots |0^{\otimes m}\rangle|0^{\otimes 2^m}\rangle$ and apply \(H^{\bigotimes(n)}\) to registers 1 and 2:
    $$\frac{1}{\sqrt{2}}(|0\rangle+|1\rangle)\frac{1}{\sqrt{2}}\sum_{u\in\{0,1\}^{n-t}}|u\rangle|0^{\otimes m}\rangle\otimes|0^{\otimes m}\rangle...\otimes|0^{\otimes2^{t}m}\rangle$$\\
    step 2: Query the oracle: 
\[\begin{aligned}
&\frac{1}{\sqrt{2}} \cdot \frac{1}{\sqrt{2}} \sum_{u \in \{0,1\}^{n-t}} \Big( |0\rangle|u\rangle |f_{w_0}(u)\rangle|f_{w_1}(u)\rangle\cdots|f_{w_{2^t-1}}(u)\rangle|0\rangle \\
&+ |1\rangle|u\rangle |g_{w_0}(u)\rangle|g_{w_1}(u)\rangle\cdots|g_{w_{2^t-1}}(u)\rangle|0\rangle \Big)
\end{aligned}
\]
step 3: According to the result after sorting the last register by
acting \(U_{sort}\) on the 3rd -\((2^{n - 1} + 2)\)th register:  
    \[\begin{aligned}
&\frac{1}{\sqrt{2}} \cdot \frac{1}{\sqrt{2}} \sum_{u \in \{0,1\}^{n-t}} \Big( |0\rangle|u\rangle |f_{w_0}(u)\rangle|f_{w_1}(u)\rangle\cdots|f_{w_{2^t-1}}(u)\rangle|S(u)\rangle \\
&+ |1\rangle|u\rangle |g_{w_0}(u)\rangle|g_{w_1}(u)\rangle\cdots|g_{w_{2^t-1}}(u)\rangle|T(u)\rangle \Big)
\end{aligned}\]  \\
    step 4: Measuring the last register, register 2 collapses to \(u_{0}\) and backs out the middle \(2^{t} - 1\) registers:
    $$\frac{1}{\sqrt{2}}(|0\rangle|u_0\rangle + |1\rangle|u_0 + a_2\rangle)$$ \\
    step 5: Apply QFT to register 2:
    $$\frac{1}{\sqrt{2^n}}(\sum_{j=0}^{2^{n-t}-1}e^{2\pi i j u_0/2^{n-t}}|0\rangle|j\rangle+\sum_{k=0}^{2^{n-t}-1}e^{2\pi ik(u_0+a_2)/2^{n-t}}|1\rangle|k\rangle)$$\\
    step 6: Measure register 2 and get \(l\), Register 1 collapses to: 
    \[
    \begin{aligned}
    |\psi_l\rangle &= \frac{1}{\sqrt{2}} \left( e^{2\pi ilu_0/2^n}|0\rangle + e^{2\pi il(u_0+a_2)/2^n}|1\rangle \right)  \frac{1}{\sqrt{2}} e^{2\pi ilu_0/2^n} \left( |0\rangle + e^{2\pi ila_2/2^n}|1\rangle \right) \\
    &= |0\rangle + e^{2\pi ila_2/2^n}|1\rangle
    \end{aligned}\]\\
    step 7: Through the sieve method to get \(|\psi_{2^{n - 2}} \rangle\) , \(l\) is equal to \(2^{n - t-1}\) , using the \(H\) door, when the last digit of \(a_{2}\) is an even number of measurements to get \(0\), The last digit i s an odd
    number of measurements to get 1.
\end{algorithm}
Similarly, by modifying the original function \({f_{w}}^{'}(u) = f_{w}(2u)\) and \({g_{w}}^{'}(u) = g_{w}(2u + a_{0})\), repeat running algorithm 3 to obtain the remaining bits of $a_2$, and then recursively obtain $a_1$ through the above formula, and finally recover all the information of $a$.

Finally, the time complexity of algorithm 3 is briefly analyzed, and the complexity of algorithm 2 is the case of \(t = 1\). The complexity of the algorithm of three also depends on two parts: one is the screening method, the time complexity of \(2^{O(\sqrt{n - t})}\), $2$ it is to get all $a_2$ bits need to be repeated \(O(n - t)\) time, two parts complexity multiplied to \(2^{O(\sqrt{n - t})} + 2^{O\left( \sqrt{n - t - 1} \right)}\ldots\ldots{+ 2}^{O\left( \sqrt{2} \right)} = 2^{O(\sqrt{n - t})}\). Similarly, to get\(a_{1}\). Similarly, the time complexity of $a_1$ is \(2^{O(\sqrt{t})} + 2^{O\left( \sqrt{t - 1} \right)}\ldots\ldots{+ 2}^{O\left( \sqrt{2} \right)} = 2^{O(\sqrt{t})}\), the overall complexity of: \(2^{O\left( \sqrt{n - t} \right)} + 2^{O(\sqrt{t})}\), when \(n \gg t\),\(2^{O(\sqrt{t})} \ll 2^{O\left( \sqrt{n - t} \right)}\), the complexity can be reduced to \(2^{O\left( \sqrt{n - t} \right)}\).

\section{Experiments}\label{sec4}
In this section, we further elucidate the correctness and effectiveness
of the algorithm by running the distributed Kuperberg’s algorithm on the
Qiskit version 0.44 platform. The functions defined in the experiments
are the one-shot functions \(f:\{ 0,1\}^{3} \rightarrow \{ 0,1\}^{4}\),  \(g:\{ 0,1\}^{3} \rightarrow \{ 0,1\}^{4}\), satisfying \(f(x) = g(x + a\ mod\ N)\), Their truth tables are shown in Table \ref{table1}.
\begin{table}[h]\label{table1}
\caption{\textbf{Truth table}}
\centering
\begin{tabular}{cccccc}
\toprule
$x$&$g(x)$&$f(x)$  & $x$& $g(x)$&$f(x)$  \\
\midrule
\(000\) & \(1001\) & \(1000\)  & \(100\) & \(0111\) &\(0101\)  
\\
\(001\) & \(1100\) & \(1001\)  & \(101\) & \(0011\) &\(0111\)  
\\
\(010\) & \(1010\) & \(1100\)  & \(110\) & \(0001\) &\(0011\)  
\\
\(011\) & \(0101\) & \(1010\)  & \(111\) & \(1000\) &\(0001\)  \\
\midrule
\end{tabular}
\end{table}
The number of experimental runs is 2048, the number of nodes in the distributed experiment is 2, the number of input qubits is 3, the number of input qubits in the original Kuperberg experiment is 4, and the parameter $a=111=7$. According to the DORCIS tool mentioned in Ref.\cite{Chun-DORCIS-EPRINT-2023}, the quantum circuit design of the function Oracle component of the traditional Kuperberg's algorithm and the distributed Kuperberg's algorithm is carried out, as shown in Fig.\ref{Figure2} and Fig.\ref{Figure3}
\begin{figure}
	\centering
	\begin{subfigure}{0.45\linewidth}
		\centering
		\includegraphics[width=0.9\linewidth]{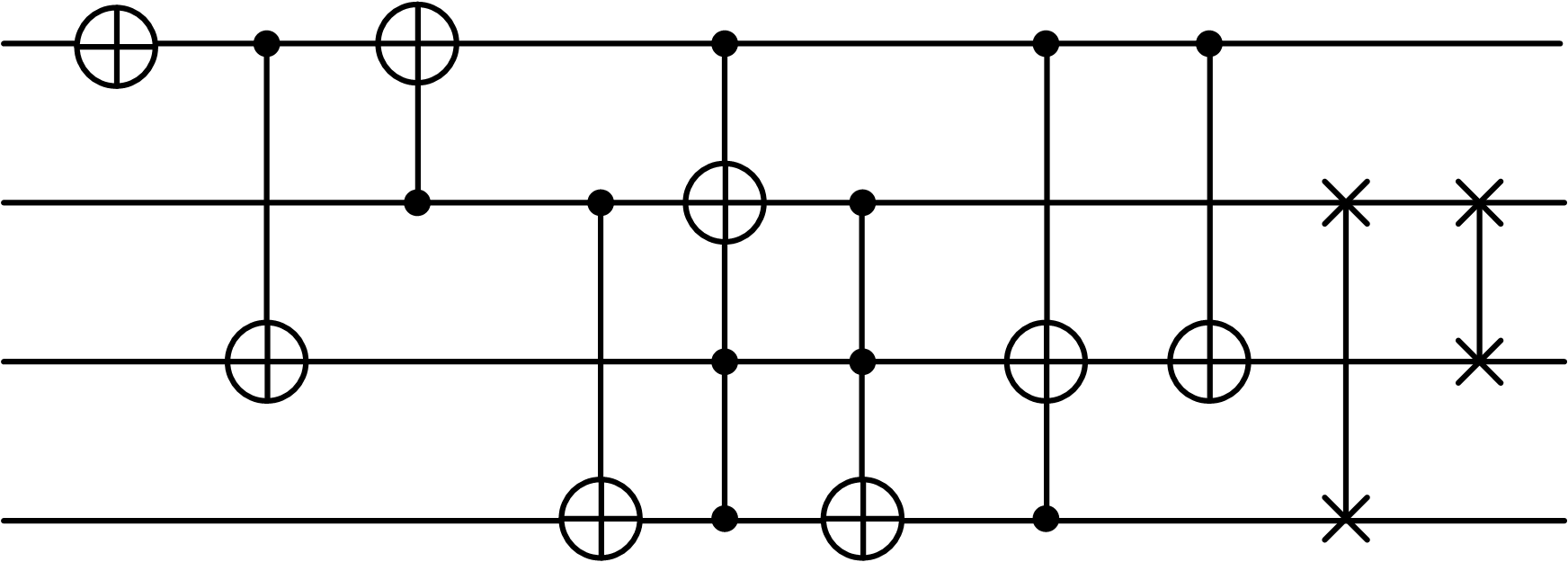}
		\caption{$O_{f}$Schematic of a quantum circuit}
		\label{chutian3}
	\end{subfigure}
	\centering
	\begin{subfigure}{0.45\linewidth}
		\centering
		\includegraphics[width=0.9\linewidth]{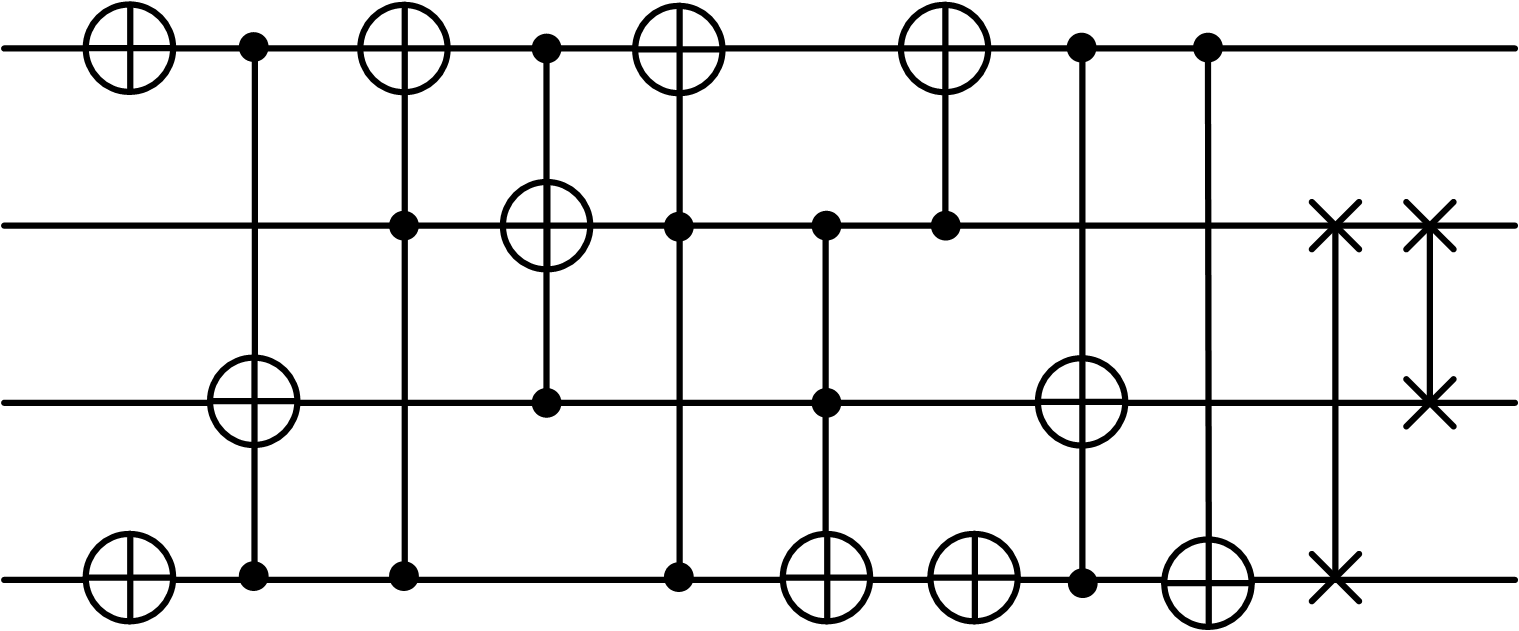}
		\caption{$O_{g}$Schematic of a quantum circuit}
		\label{chutian3}
	\end{subfigure}
	\caption{Schematic diagram of Oracle quantum circuit of original Kuperberg's algorithm}\label{Figure2}
	\label{da_chutian}
\end{figure}
\begin{figure}[h]
	\centering
	\begin{subfigure}{0.45\linewidth}
		\centering
		\includegraphics[width=0.75\linewidth]{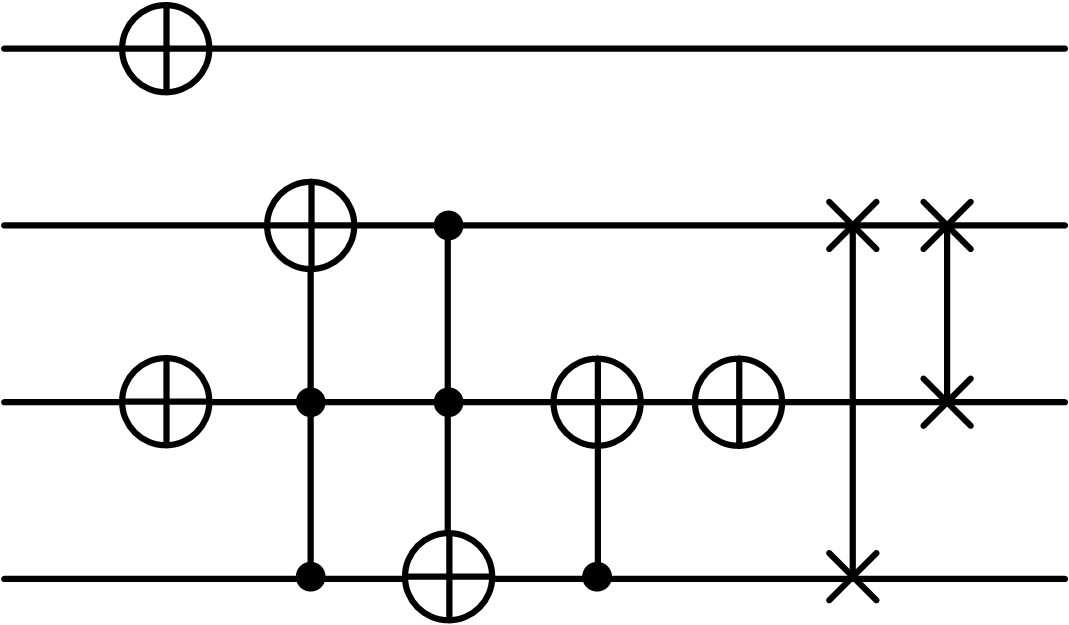}
		\caption{$O_{f_0}$Schematic of a quantum circuit}
		\label{chutian3}
	\end{subfigure}
	\centering
	\begin{subfigure}{0.45\linewidth}
		\centering
		\includegraphics[width=0.5\linewidth]{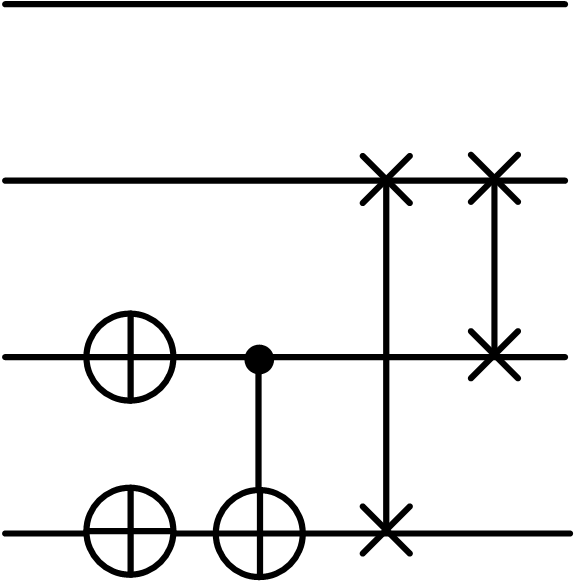}
		\caption{$O_{f_1}$Schematic of a quantum circuit}
		\label{chutian3}
	\end{subfigure}
    \qquad
    \centering
	\begin{subfigure}{0.45\linewidth}
		\centering
		\includegraphics[width=0.9\linewidth]{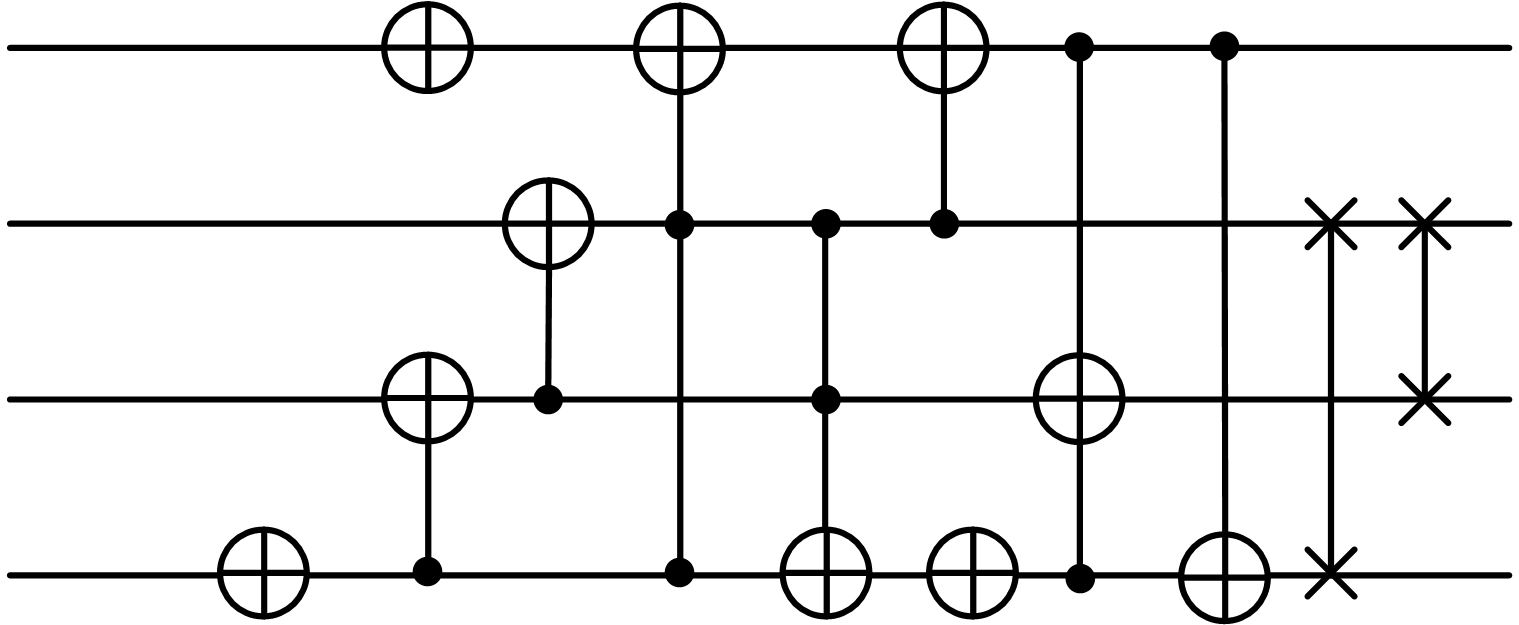}
		\caption{$O_{g_0}$Schematic of a quantum circuit}
		\label{chutian3}
	\end{subfigure}
	\centering
	\begin{subfigure}{0.45\linewidth}
		\centering
		\includegraphics[width=0.9\linewidth]{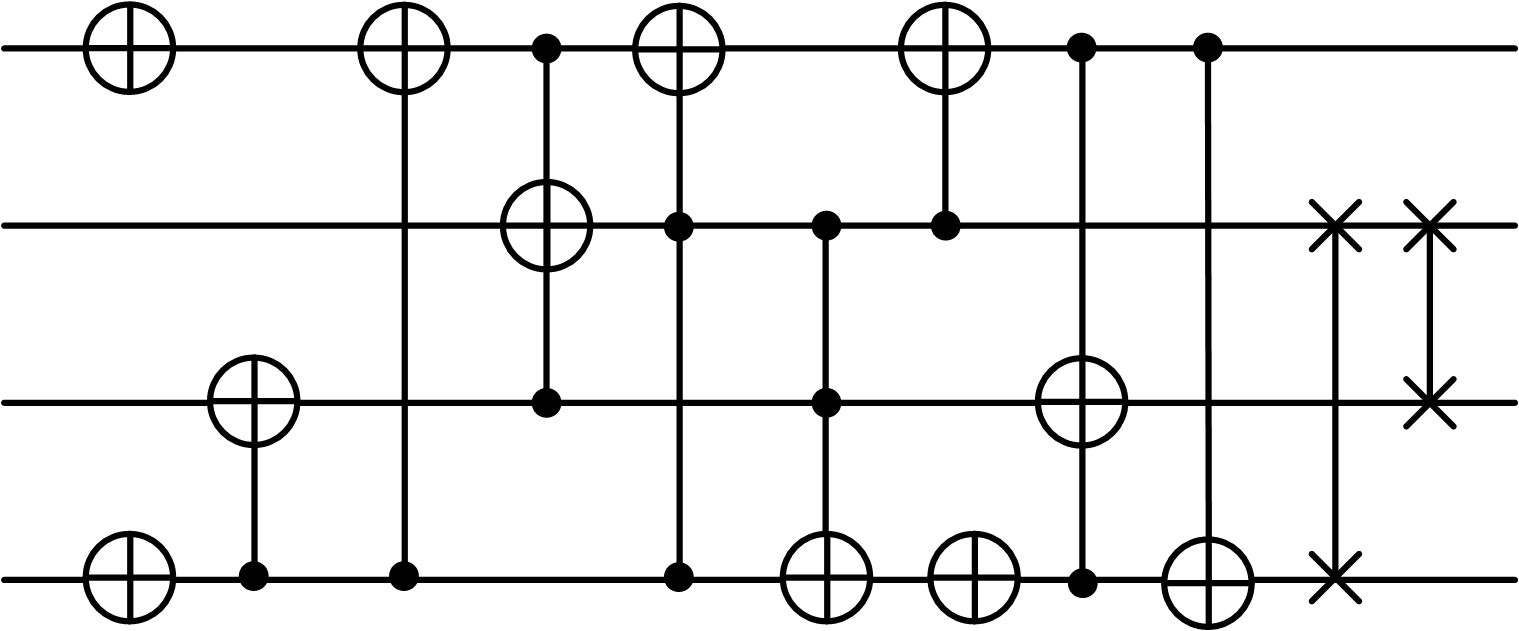}
		\caption{$O_{g_1}$Schematic of a quantum circuit}
		\label{chutian3}
	\end{subfigure}
	\caption{Schematic diagram of distributed Kuperberg's algorithm Oracle quantum circuits}
	\label{Figure3}
\end{figure}
The experimental procedure includes superposition state initialization,
Oracle query, quantum sorting \(U_{sort}\) application and QFT
measurement. The experimental results are shown in Figs.\ref{Figure4}, which indicate that the distributed version has a measurement success rate of
22.6\% in obtaining the same results, which is significantly higher than
that of the traditional algorithm, which is 10.6\%. The circuit depth of
a single node of the distributed architecture is reduced by 22\% on
average. The experiments verify the correctness of the distributed
Kuperberg algorithm, which has a greater advantage in reducing the
resource requirements and improving the probability of success.
\begin{figure}[h]
	\centering
	\begin{subfigure}{0.45\linewidth}
		\centering
		\includegraphics[width=0.9\linewidth]{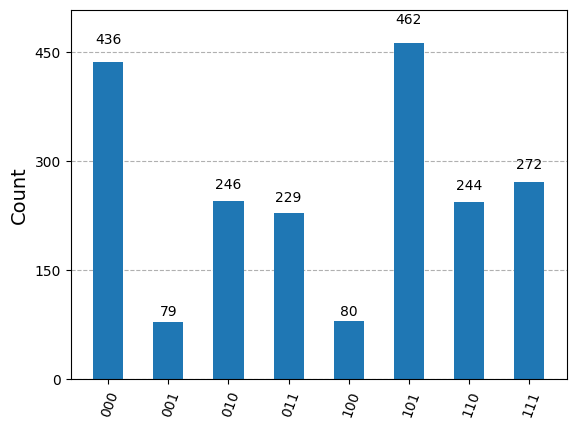}
		\caption{Distributed Kuperberg's algorithm with input bit 3 and correct result 101}
		\label{chutian3}
	\end{subfigure}
	\centering
	\begin{subfigure}{0.45\linewidth}
		\centering
		\includegraphics[width=0.9\linewidth]{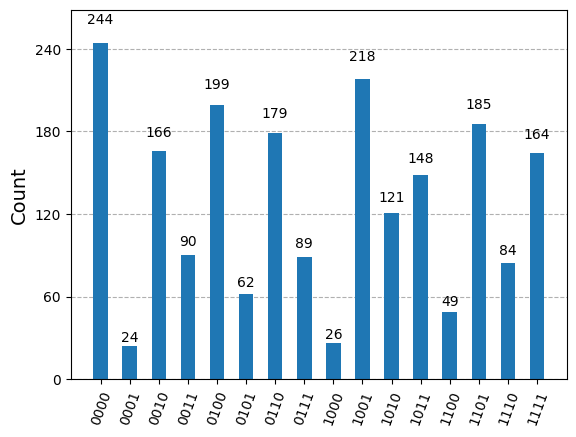}
		\caption{ Kuperberg's
algorithm with input bit 4 and correct result 1001}
		\label{chutian3}
	\end{subfigure}
	\caption{Comparison of experimental results}
	\label{da_chutian}
\end{figure}

\section{Conclusion}\label{sec5}
In this paper, a novel Kuperberg algorithm for distributed quantum computing environments is proposed, aiming to solve the problem of excessive quantum resource demand of the traditional algorithm. By decomposing the original function into multiple subfunctions and assigning them to different quantum nodes for parallel processing, the algorithm significantly reduces the depth of the quantum circuit and the number of qubits in a single node, while optimizing the time complexity. Theoretical analysis shows that the distributed version reduces the complexity from \(2^{O\left( \sqrt{n} \right)}\) to \(2^{O\left( \sqrt{n - t} \right)}\)  ($t$ is the number of nodes), and the advantage is especially obvious when the number of nodes is much smaller than the problem size. The experimental part validates the accuracy of the algorithm on the Qiskit platform, and the results show that the distributed version not only reduces the depth of the circuit, but also reduces the noise impact through parallelization. Additionally, the probability of measurement success is higher than that of the traditional method. This work provides new ideas for efficiently solving DHSP on NISQ devices. Future research can further explore more flexible function decomposition strategies and cross-node communication optimization to enhance the scalability and practicality of the algorithm.

\bibliography{ref}

\end{document}